\documentclass[a4paper,USenglish,cleveref,autoref,thm-restate]{lipics-v2021}

\pdfoutput=1
\hideLIPIcs
\usepackage[utf8]{inputenc}
\usepackage{hyperref}
\usepackage{graphicx}
\usepackage{amsmath}
\usepackage{amsthm}
\usepackage{amssymb}
\usepackage{dsfont}
\usepackage{cleveref}
\usepackage[dvipsnames]{xcolor}
\usepackage{float}
\usepackage{tikz}
\usetikzlibrary{shapes.geometric} 
\usepackage{xspace}
\usepackage[normalem]{ulem}

\DeclareFontShape{T1}{lmr}{m}{scit}{<->ssub*lmr/m/scsl}{}
\DeclareFontShape{T1}{lmr}{bx}{scit}{<->ssub*lmr/m/scsl}{}

\definecolor{blue}{RGB}{0,50,200}
\definecolor{magenta}{RGB}{255,0,255}
\hypersetup{colorlinks,linkcolor=blue,urlcolor=blue,citecolor=magenta}

\title{Online Bisection with Ring Demands}

\author{Mateusz Basiak}{University of Wrocław, Poland}{mateusz.basiak@cs.uni.wroc.pl}{https://orcid.org/0009-0009-0210-6451}{}
\author{Marcin Bienkowski}{University of Wrocław, Poland}{marcin.bienkowski@cs.uni.wroc.pl}{https://orcid.org/0000-0002-2453-7772}{}
\author{Guy Even}{
School of Electrical Engineering, Tel-Aviv University, Tel Aviv, Israel \and
Max Planck Institute for Informatics, Saarbrücken, Germany}{guyeven@mpi-inf.mpg.de}{https://orcid.org/0000-0001-5407-330X}{}
\author{Agnieszka Tatarczuk}{University of Wrocław, Poland}{agnieszka.tatarczuk@cs.uni.wroc.pl}{https://orcid.org/0009-0008-3849-9165}{}

\authorrunning{M. Basiak, M. Bienkowski, G. Even and A. Tatarczuk}
\Copyright{Mateusz Basiak, Marcin Bienkowski, Guy Even and Agnieszka Tatarczuk} 

\funding{Supported by Polish National Science Centre 
\bigskip
grant 2022/45/B/ST6/00559.\\
\emph{For the purpose of open access, the authors have applied CC-BY public
copyright license to any Author Accepted Manuscript (AAM) version arising from
this submission.}}

\nolinenumbers 

\newcommand{\HC}[1]{#1_{\textrm{HIT}}}
\newcommand{\MC}[1]{#1_{\textrm{REC}}}
\newcommand{\OPT}{\textsc{Opt}\xspace}
\newcommand{\OFF}{\textsc{Off}\xspace}
\newcommand{\ALG}{\textsc{Alg}\xspace}

\newcommand{\OB}{\textsc{Onl}\xspace}
\newcommand{\CE}[1]{\textsc{cut}(#1)}
\newcommand{\eps}{\varepsilon}
\newcommand{\AC}{A^{\circ}}
\newcommand{\A}{\mathcal{A}}
\newcommand{\C}{\mathcal{C}}
\newcommand{\E}{\mathbf{E}}
\newcommand{\tp}{\tilde{p}}
\newcommand{\less}{\textsc{less}}

\begin{document}

\maketitle

\begin{abstract}
    The online bisection problem requires maintaining a dynamic partition of $n$ nodes into two equal-sized clusters. Requests arrive sequentially as node pairs. If the nodes lie in different clusters, the algorithm pays unit cost. After each request, the algorithm may migrate nodes between clusters at unit cost per node. This problem models datacenter resource allocation where virtual machines must be assigned to servers, balancing communication costs against migration overhead.

    We study the variant where requests are restricted to edges of a ring network, an abstraction of ring-allreduce patterns in distributed machine learning. Despite this restriction, the problem remains challenging with an $\Omega(n)$ deterministic lower bound. We present a randomized algorithm achieving $O(\varepsilon^{-3} \cdot \log^2 n)$ competitive ratio using resource augmentation that allows clusters of size at most~$(3/4 + \varepsilon) \cdot n$.

    Our approach formulates the problem as a metrical task system with a restricted state space. By limiting the number of cut-edges (i.e., ring edges between clusters) to at most $2k$, where $k = \Theta(1/\varepsilon)$, we reduce the state space from exponential to polynomial (i.e., $n^{O(k)}$). The key technical contribution is proving that this restriction increases cost by only a factor of $O(k)$. Our algorithm follows by applying the randomized MTS solution of Bubeck et~al.~[SODA 2019].

    The best result to date for bisection with ring demands is the $O(n \cdot \log n)$-competitive deterministic online algorithm of Rajaraman and Wasim [ESA 2024] for the general setting. While prior work for ring-demands by Räcke et~al. [SPAA 2023] achieved $O(\log^3 n)$ for multiple clusters, their approach employs a resource augmentation factor of $2+\varepsilon$, making it inapplicable to bisection. 
\end{abstract}


\section{Introduction}

This work is motivated by the increasing interest in online graph partitioning problems. Among these, one of the most elegant and succinctly defined problems is the \emph{online bisection problem}~\cite{AvLoPS16}. The goal is to maintain a time-varying partition of $n$ nodes into two clusters, each of capacity $n/2$. The input is a sequence of requests, each corresponding to a pair of nodes. Serving a~pair whose nodes are in different clusters incurs a unit cost, while serving requests between nodes in the same cluster is free. After serving a request, the algorithm may update the partition, paying a unit cost for each node that changes clusters.

The objective is to minimize the total incurred cost. This is an \emph{online problem}: requests arrive sequentially, and the algorithm has to serve each before observing the next. The performance of an algorithm is measured by its competitive ratio~\cite{BorEl-98}, defined as the worst-case ratio between its total cost and that of an optimal offline algorithm \OPT that knows the entire sequence in advance.

\subparagraph{Motivation.}

Node pairs are ephemeral: once a request is served, it no longer incurs any cost. The online bisection problem models trade-offs that naturally arise in datacenters, where virtual machines (nodes) exchanging data have to be assigned to racks or physical servers (clusters). When two virtual machines are colocated (in the same cluster), their communication is cost-free; otherwise, it incurs latency and consumes bandwidth. Modern virtualization enables migration of virtual machines between servers (i.e., changes of node clusters), but such migrations incur network cost of their own.


\subsection{Previous results.}

This seemingly simple problem has proved surprisingly challenging, and despite substantial effort over the past decade, our understanding remains limited.

\subparagraph{Exact online bisection.}

The first algorithm, proposed by Avin et al.~\cite{AvLoPS16,AvBLPS20}, was a simple deterministic $O(n^2)$-competitive algorithm that maintained connected components of nodes that had communicated so far, mapping each component to a single cluster. This idea was later refined by Bienkowski and Schmid~\cite{BieSch24}, who proved that a slightly randomized version of this approach achieves a sub-quadratic competitive ratio of $\widetilde{O}(n^{2-1/{12}})$.

A polynomial dependence on $n$ is unavoidable for deterministic algorithms: via a reduction from online paging~\cite{SleTar85}, one obtains a lower bound of $\Omega(n)$~\cite{AvBLPS20}. For randomized algorithms, however, the known lower bound is merely $\Omega(\log n)$~\cite{HeNeRS21}, leaving an exponential gap between the best upper and lower bounds.\footnote{The lower bound of~\cite{HeNeRS21} was established for the so-called learning-model, but repeating their construction multiple times yields a lower bound for the general variant of the problem.}

\subparagraph{Resource augmentation.}

The limited progress on exact bisection prompted the exploration of \emph{resource-augmented} variants, where the online algorithm is permitted to relax the requirement of equal-sized clusters. More precisely, a \emph{$(1+\eps)$-augmented} algorithm allows at most $(1+\eps) \cdot (n/2)$ nodes in each cluster, while being compared against a non-augmented \OPT with clusters of equal size~$n/2$. We refer to $1+\eps$ as the \emph{augmentation factor} or \emph{balance parameter} of the algorithm. This relaxation is well motivated in practice because servers often have spare capacity that can be exploited to improve performance.

Surprisingly, the lower bound for deterministic algorithms in this setting remains $\Omega(n)$, and it holds for any nontrivial amount of augmentation, as long as it is not possible to place all nodes in a single cluster~\cite{AvBLPS20}. On the positive side, Rajaraman and Wasim~\cite{RajWas22} presented an $O(n \log n)$-competitive deterministic algorithm for any~fixed~$\eps > 0$.\footnote{Their algorithm also applies to the multi-cluster variant and uses an approach of Henzinger et al.~\cite{HeNeRS21} as~a~subroutine, thereby inheriting an exponential dependence on $1/\eps$ hidden in the $O$-notation. However, for the online bisection, a straightforward extension of the component-based algorithm of~\cite{AvBLPS20} to the resource-augmented setting would trivially yield an $O(\eps^{-1} \cdot n \log n)$-competitive algorithm.} No randomized algorithm achieving a better bound than $O(n \log n)$ is known.

\subparagraph{Multiple clusters.}

The online bisection problem has also been studied in a generalized form, known as \emph{online balanced graph partitioning}, where there are $\ell \geq 2$ clusters, each of size~$n / \ell$~\cite{AvBLPS20,AvLoPS16,BBKRSV21,FoRSc21,PaPaSc21,RScZa22,RajWas22}. This generalization has also been investigated in models with a~large augmentation factor of \mbox{$1+\eps \geq 2 $}~\cite{AvBLPS20,AvLoPS16,FoRSc21,RScZa22}, which would trivialize the online bisection problem, as all nodes could then be placed in a single cluster.

\subparagraph{Ring demands.}

Another natural restriction is to assume that requests must belong to a~predefined subset $Q \subseteq V \times V$ that is known \emph{a priori} to the online algorithm.

The case when $Q$ is a cycle connecting all nodes is of particular theoretical and practical interest. In fact, the ring-allreduce communication pattern~\cite{SerBal18} used in distributed machine learning follows this structure. For the theoretical perspective, the deterministic lower bound of~$\Omega(n)$ holds even under ring demands~\cite{AvLoPS16,AvBLPS20}. Ring demands were also studied by Räcke et al.~in the multi-cluster case, for which they presented an $O(\log^3 n)$-competitive randomized algorithm~\cite{RScZa23}. However, their algorithm crucially requires an augmentation factor $2+\eps$, which makes it inapplicable to the online bisection problem.


\subsection{Our Contribution}

In this work, we present a randomized online algorithm for the online bisection problem with ring demands. The augmentation factor is $3/2 +\varepsilon$, for $\varepsilon\in(0,1/2]$,\footnote{ In other words, each cluster may contain at most $(3/4+\varepsilon/2) \cdot n$ nodes.} and the competitive ratio is $O(\varepsilon^{-3} \cdot \log^2 n)$. For this setting, the best known solution to date is the $(1+\eps)$-augmented, $O(n  \log n)$-competitive deterministic algorithm by~Rajaraman and Wasim~\cite{RajWas22}, which remains the state of the art even when randomization and arbitrarily high augmentation are permitted.

\subparagraph{Overview.}

The bisection problem can be viewed as a \emph{metrical task system} (MTS)~\cite{BoLiSa92} in which every state is a partition, and distances between states are the ``edit distances'' between these partitions. Let $\mathcal{S}$ denote the set of states. The randomized online MTS algorithm of Bubeck et~al.~\cite{BuCoLL19} yields a competitive ratio $O(\log^2 |\mathcal{S}|)$, which is $O(n^2)$ since $|S|=\binom{n}{n/2}$.

To reduce the competitive ratio, we restrict the number of states to $n^{O(k)}$, for $k=\Theta(1/\varepsilon)$, as follows. Any partition of the cycle into two clusters can be represented by an even set of \emph{cut-edges}, meaning ring edges between nodes assigned to different clusters. A~natural choice for restricting the state space is to consider sets of at most $2k$ cut-edges. 

More formally, let MTS$(2k,\alpha)$ denote the MTS in which every state is an $\alpha$-balanced partition induced by at most $2k$ cut-edges. As the number of states of this MTS is $O(n^{O(k)})$, the randomized online MTS algorithm of Bubeck et al.~\cite{BuCoLL19} has cost at most $O(k^2 \cdot \log^2 n)$ times that of an \emph{optimal offline} $\alpha$-balanced algorithm using at most $2k$ cut-edges. 

This approach shifts the challenge to bounding the increase in cost caused by this sparsification of the set of states, i.e., bounding the ratio between the costs of the optimal solutions for MTS$(n,1)$ and MTS$(2k,\alpha)$. We succeed in bounding this ratio by $O(k)$ if $\alpha=3/2 + 1/k$ (see~\autoref{thm:offline}).
 
Our analysis works in phases in which an offline solution $\OFF$ for MTS$(2k,\alpha)$ (for $\alpha=3/2+1/k$) ``chases'' an optimal solution $\OPT$ for MTS$(n,1)$ until the $\alpha$-balance parameter is violated. At all times, we maintain the invariant that the cut-edges in $\OFF$ form a subset of the cut-edges of $\OPT$. This invariant guarantees that a request incurs a cost on $\OFF$ only if it incurs cost on $\OPT$. 

At the beginning of each phase, we start by sparsifying the cut-edges of $\OPT$ to a subset of cardinality at most $2k$. \autoref{lem:rebalancing} proves that this is possible with balance parameter of $1+1/k$. The phase proceeds with $\OFF$ iteratively responding to changes in $\OPT$ and ends when this response leads to a partition that is no longer $\alpha$-balanced.

We show that changes in $\OFF$ (in response to the single-node changes of \OPT) incur amortized cost that is bounded by the cost of \OPT within the phase (see \autoref{cor:step}). Moreover, the $O(n)$ rebalancing cost per phase can be charged to the cost of \OPT within the phase, which we lower-bound by $\Omega(n/k)$ (see~\autoref{lem:opt_lb}).

\subparagraph{Comparison to prior work.} 

Superficially, our algorithm shares certain similarities (the notion of cut-edges and the use of MTS routines) with the algorithm of Räcke et al.~\cite{RScZa23}, but the underlying mechanisms differ. In~particular, in their approach, the number of cut-edges corresponds to the number of~clusters, whereas in our case it serves as a parameter that affects both the cost and the augmentation factor. Moreover, the partition of the cycle in their approach is more rigid, since cut-edges are confined to specific segments of the cycle that remain fixed throughout the execution, whereas in our approach the partition is more flexible, allowing cut-edges to appear anywhere on the cycle.


\subsection{Other Related Work}

The online bisection problem and its generalization to the online balanced partitioning problem were also studied in several relaxed variants, such as the learning model (where the input admits a partition incurring no service cost)~\cite{HeNeRS21,HeNeSc19,PaPaSc20,PaPaSc21,RScZa23,RScZa25} and stochastic inputs~\cite{AvCoPS19}.

The offline counterparts of these problems (bisection and $k$-balanced partitioning) are NP-hard, and their approximation ratios have been improved in a long line of works; see~\cite{SarVaz95,ArKaKa99,FeKrNi00,FeiKra02,KraFei06,Raec08} for bisection and~\cite{SimTen97,KrNaSc09,AndRae06,EvNaRS99} for $k$-balanced partitioning. Recently, balanced partitioning was also studied in a setting where the input and output are as in the online model, but the algorithm operates offline and has access to the entire input~\cite{RScZa22}.

Closely related online problems include online disengagement (a ``dual'' of online bisection, where the algorithm pays for requests whose endpoints lie in the same cluster)~\cite{AzMaPT23,PaRoUm25,RajWas24} and the minimum linear arrangement problem, where the goal is to maintain a dynamically changing linear ordering of nodes~\cite{BieEve24,OlPSSS18}.


\subsection{Problem Definition}

In the online bisection problem, we are given $n$ nodes $v_1, v_2, \dots, v_n$, where $n$ is even. Instead of defining the problem in terms of partitions, we use an equivalent but more convenient coloring formulation: the algorithm maintains a coloring of all nodes with two colors, red and blue. There is a fixed initial coloring, where half of the nodes are colored red and the other half blue. 

Requests arrive online. Each request is a pair of nodes, referred to as its endpoints. For every request in an input sequence:
\begin{itemize}
    \item if its endpoints have the same color, the algorithm pays a \emph{hit cost} of $0$;
    \item if its endpoints have different colors, the algorithm pays a \emph{hit cost} of $1$;
    \item after paying the hit cost, the algorithm may \emph{recolor} any number of nodes, paying a~\emph{recoloring cost} equal to the number of recolored nodes.
\end{itemize}
For an algorithm $\ALG$ and input $\sigma$, we write $\ALG(\sigma)$ for the total cost on $\sigma$. We denote by $\HC{\ALG}$ and $\MC{\ALG}$ the \emph{hit cost} and \emph{recoloring cost}, respectively.

\subparagraph{Resource augmentation.}

A coloring is \emph{$\alpha$-balanced} (for $\alpha \geq 1$) if at most $\alpha \cdot (n/2)$ nodes are assigned each color. An algorithm is \emph{$\alpha$-augmented} if, once it performs the recoloring in~response to a request, the resulting coloring is $\alpha$-balanced.

An ($\alpha$-augmented) algorithm \ALG is \emph{$R$-competitive} if there exists $\beta$ such that, for every input $\sigma$, it holds that $\ALG(\sigma) \leq R \cdot \OPT(\sigma) + \beta$, where \OPT denotes the optimal offline algorithm that maintains a $1$-balanced coloring. The parameter $\beta$ cannot depend on $\sigma$, but may be a function of $n$. 
For randomized algorithms, we replace $\ALG(\sigma)$ above by its expected value and assume standard, oblivious adversary~\cite{BorEl-98} that does not see random bits used by the algorithm.

\subparagraph{Ring demands.}

In the variant studied in this paper, all requests correspond to edges of~a~cycle. Thus, there are $n$ possible request types: $(v_1, v_2), (v_2, v_3), \dots, (v_n, v_1)$.


\section{Outline of Our Solution}
\label{sec:ourresults}

We show that for every $\eps \in (0, 1/2]$, there exists an online $(3/2 + \eps)$-augmented randomized algorithm that is $O(\eps^{-3} \cdot \log^2 n)$-competitive. Below we outline our approach, starting with several necessary definitions.

\subparagraph{Cut-edges.}

At any time, the coloring maintained by algorithm $\ALG$ induces a set of \emph{cut-edges}, denoted $\CE{\ALG}$, consisting of edges of the cycle whose endpoints have different colors. Note that $|\CE{\ALG}|$ is even. Conversely, each even-cardinality set $\C$ of cut-edges corresponds to two equivalent colorings of the nodes, obtained from each other by swapping color names. A set of cut-edges is \emph{valid} if its cardinality is even. Hence, a valid set $\C$ of~cut-edges uniquely represents a coloring, up to a swap of color names. 

Fix any valid set $\C$ of cut-edges. The set $\C$ partitions the cycle into monochromatic arcs of alternating colors. Each such arc consists entirely of red or blue nodes, and we refer to it as a red or blue arc, respectively. We also define other arcs: for edges $e_i$ and $e_j$, let $A(e_i, e_j)$ denote the set of nodes on the shorter path between $e_i$ and $e_j$ along the cycle. Let $d(e_i, e_j) \triangleq |A(e_i, e_j)|$.

\subparagraph{Color frequency.}

For a given coloring, a color is called \emph{more frequent} if the number of nodes of that color exceeds $n/2$, and \emph{less frequent} otherwise. When both colors contain exactly $n/2$ nodes, either color may be treated as more frequent. For a valid set $\C$ of cut-edges, let $\less(\C)$ denote the number of nodes assigned the less frequent color in the coloring induced by $\C$; cf.~\autoref{fig:coloring}.

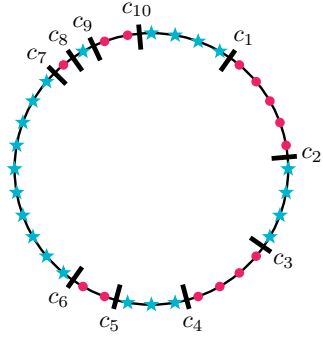
\begin{figure*}[t!]
    \centering
    \scalebox{0.9}{
         \begin{tikzpicture}[line cap=rect,line width=1pt]
        \filldraw [fill=white] (0,0) circle [radius=2cm];
       
        \foreach \angle [count=\xi] in {55, 5, -35, -75, -105, -125, -225, -235, -245, -265}
        {
          \draw[line width=2pt, color=black] (\angle:1.8cm) -- (\angle:2.1cm);
          \node[font=\normalsize] at (\angle:2.36cm) {$c_{\xi}$};
        }
         \foreach \k in { 10, 20, 30, 40, 50,  100, 110,  130,  240, 250,  290, 300, 310, 320}{
            \fill[WildStrawberry]  ++(\k:2cm) circle[radius=2pt];
         }
         \foreach \k in { 60, 70, 80, 90,  120,  140, 150, 160, 170, 180, 190, 200, 210, 220, 230,  260, 270, 280,  330, 340, 350, 360}{
            \node[star, star points=5, star point ratio=2.5, draw=Turquoise, fill=Turquoise, minimum size=4pt, inner sep=0pt] at (\k:2cm) {};
  
         }

        \end{tikzpicture}
    }
    \caption{An example coloring of a cycle with the set of cut-edges $\C = {c_1, c_2, \dots, c_{10}}$. The cycle contains 22 blue and 14 red nodes; hence the less frequent color is red, and $\less(\C) = 14$.}
    \label{fig:coloring}
\end{figure*}

\subparagraph{Global rebalancing.}

As stated in the introduction, we consider algorithms that use at most $2k$ cut-edges, for an integer constant $k$. However, although the initial coloring is $1$-balanced, it may involve more than $2k$ cut-edges. To address this, we introduce a \emph{global rebalancing} procedure; its properties are stated in the lemma below and proved in \autoref{sec:global_rebalancing}. The procedure is executed at the beginning and may also be invoked later during the execution of the algorithm.

\begin{restatable}[Global rebalancing procedure]{lemma}{rebalancing}
    \label{lem:rebalancing}
    For any valid $1$-balanced set $\C$ of cut-edges, there exists a deterministic procedure that computes a subset $\C_{2k} \subseteq \C$ of size $\min\{2k, |\C|\}$ such that $\C_{2k}$ is $(1+1/k)$-balanced.
\end{restatable}

For every integer $k \geq 1$ and $\alpha \geq 1 + 1/k$, define the class $\A^{k}_{\alpha}$ of algorithms that (i) are $\alpha$-augmented, (ii) perform global rebalancing, as specified in \autoref{lem:rebalancing}, before processing the input sequence, and (iii) during execution, maintain at most $2k$ cut-edges. Algorithms in $\A^{k}_{\alpha}$ may be either offline or online.

\subparagraph{Main bounds and paper organization.}

The main technical contribution of this paper is the following theorem on \emph{offline} algorithms.

\begin{restatable}{theorem}{offline}
    \label{thm:offline}
    Fix an integer $k \geq 1$ and let $\alpha \triangleq 3/2 + 1/k$. There exists an offline algorithm $\OFF \in \A^{k}_{\alpha}$ such that, for every input $\sigma$, it holds that $\OFF(\sigma) \leq O(k) \cdot \OPT(\sigma) + O(n)$.
\end{restatable}

The theorem guarantees the \emph{existence} of an algorithm~$\OFF$; the algorithm need not be efficient. \autoref{sec:offline} presents its construction and analysis.

Next, we show that within the class~$\A^{k}_{\alpha}$, one can construct an online randomized algorithm that is $O(k^2 \cdot \log^2 n)$-competitive against \emph{any offline algorithm in the same class}. The proof proceeds by formulating the problem as a metrical task system~\cite{BoLiSa92} and applying the algorithm of Bubeck et~al.~\cite{BuCoLL19}. It is given in \autoref{sec:online}.

\begin{restatable}{theorem}{online}
    \label{thm:online} 
    Fix an integer $k \geq 1$ and $\alpha \geq 1 + 1/k$. It is possible to construct an~online randomized algorithm $\OB \in \A^k_\alpha$ such that, for a fixed~$\beta$, it holds that $\E[\OB(\sigma)] \leq O(k^2 \cdot \log^2 n) \cdot \OFF(\sigma) + \beta$, 
    for every offline algorithm $\OFF \in \A^k_\alpha$ and every input~$\sigma$. 
\end{restatable}

\begin{corollary}
    \label{col:main}
    For any $\eps \in (0,1/2]$, there exists a $(3/2 + \eps)$-augmented and $O(\eps^{-3} \cdot \log^2 n)$-competitive online randomized algorithm for the online bisection problem on $n$ nodes, where all requests correspond to edges of a cycle connecting all nodes.
\end{corollary}

\begin{proof}
    Let $k \triangleq \lceil 1/\eps \rceil$ and $\alpha \triangleq 3/2 + 1/k \leq 3/2 + \eps$. Consider the online randomized algorithm $\OB \in \A^k_\alpha$ guaranteed by 
    \autoref{thm:online}. 

    Since $\OB \in \A^k_\alpha$, it is $\alpha$-augmented, and thus also $(3/2 + \eps)$-augmented. Fix an input $\sigma$ and the offline algorithm $\OFF \in \A^k_\alpha$ provided by \autoref{thm:offline}. Let $\beta$ be the parameter from the statement of \autoref{thm:online}. Then, for every input $\sigma$, 
    \begin{align*}
        \E[\OB(\sigma)] 
            &\leq O(k^2 \cdot \log^2 n) \cdot \OFF(\sigma) + \beta 
            && \text{(by \autoref{thm:online})}\\
            &\leq O(k^3 \cdot \log^2 n) \cdot \OPT(\sigma) + O(n \cdot k^2 \cdot \log^2 n) + \beta
            && \text{(by \autoref{thm:offline})}\\
            &\leq O(\eps^{-3} \cdot \log^2 n) \cdot \OPT(\sigma) + O(\eps^{-2} \cdot n \cdot \log^2 n + \beta).
    \end{align*}
    This completes the proof of the claimed competitive ratio.
\end{proof}


\section{Offline Algorithm}
\label{sec:offline}

In this section, we fix an integer $k \geq 1$, set $\alpha \triangleq 3/2+1/k$, and we construct an offline algorithm $\OFF \in \A^k_\alpha$ that generates a solution on the basis of actions of \OPT.\footnote{Although the algorithm is well defined for any integer $k \geq 1$, only values $k \geq 3$ yield non-trivial augmentation factors.} Before processing the input, \OFF applies \autoref{lem:rebalancing} to compute a coloring that uses at most $2k$ edges derived from the initial coloring of \OPT.

\begin{definition}
    \label{def:cut-invariants} An offline algorithm \OFF maintains \emph{cut-edge invariants} if
    \begin{itemize}
        \item $\CE{\OFF} \subseteq \CE{\OPT}$ and
        \item $|\CE{\OFF}| = \min \{2k,|\CE{\OPT}|\}$.
    \end{itemize}
\end{definition}

Note that the cut-invariants imply that $\CE{\OFF} = \CE{\OPT}$ if $\CE{\OPT} \leq 2k$.
By \autoref{lem:rebalancing}, the cut-edge invariants hold after the global rebalancing procedure is applied at the beginning. We will also show that \OFF maintains these invariants throughout execution.

As $\CE{\OFF} \subseteq \CE{\OPT}$, $\OFF$ never incurs hit cost unless \OPT does. Hence, $\HC{\OFF}(\sigma) \leq \HC{\OPT}(\sigma)$ for every input $\sigma$. Therefore, our goal is to construct \OFF such that $\MC{\OFF}(\sigma) \leq O(k) \cdot \MC{\OPT}(\sigma) + O(n)$ for every input $\sigma$.

\subparagraph{Steps.}

Recall that in response to a single request, \OPT may recolor multiple nodes, and after these changes are executed, the coloring of \OPT has to be $1$-balanced. In~\autoref{sec:single_recoloring}, we define the changes that \OFF performs to mimic a single-node recoloring by \OPT. Each such change maintains the cut-edge invariants (cf.~\autoref{def:cut-invariants}), but does not necessarily preserve color balance. Thus, when \OFF processes all node recolorings of \OPT in response to a single request, it verifies whether its coloring is $\alpha$-balanced. If this is not the case, \OFF executes the global rebalancing procedure described in \autoref{lem:rebalancing}. This procedure produces an $\alpha$-balanced coloring and preserves the cut-edge invariants.

For succinctness, we number the node recolorings performed by \OPT and assume that each such recoloring corresponds to a single \emph{step}, while a single request may correspond to multiple recoloring steps. We use superscript $t$ to indicate that the state (or cost) refers to step $t$, and omit it when the considered step is clear from the context.


\subsection{Global Rebalancing}
\label{sec:global_rebalancing}

In this section, we prove \autoref{lem:rebalancing} by presenting a procedure that selects a subset of cut-edges from $\CE{\OPT}$ guaranteeing the cut-edge invariants (cf.~\autoref{def:cut-invariants}). Note that this procedure disregards the current set of cut-edges of \OFF. The construction begins with the cut-edges of \OPT and removes them iteratively until the desired number of cut-edges is reached, controlling the color imbalance introduced at each step.

\begin{proof}[Proof of~\autoref{lem:rebalancing}]
    If $|\C| \leq 2k$, we simply set $\C_{2k} = \C$ and the claim follows since $\less(\C_{2k}) = n/2$. Therefore, in what follows, we assume that $|\C| = 2 m > 2k$.

    We now describe how to iteratively construct the sets $\C_{2m}, \C_{2m-2}, \dots, \C_{2k+2}, \C_{2k}$. For every $j \in \{k, \dots, m\}$, we maintain the invariant $\less(\C_{2j}) \geq n/2 - n/(2j)$. The final set $\C_{2k}$ is then $(1 + 1/k)$-balanced as required.

    We start with $\C_{2m} \triangleq \C$. Clearly, the invariant holds since $\less(\C_{2m}) = n/2$.

    It now suffices to describe how to choose $\C_{2j-2}$ based on $\C_{2j}$ (for $m \geq j > k$), so that if the invariant holds for $\C_{2j}$, it also holds for $\C_{2j-2}$. Let $c_0, c_1, \dots, c_{2j-1}$ denote the cut-edges of~$\C_{2j}$ in clockwise order, starting from an arbitrary one. These cut-edges are numbered modulo $2j$, that is, $c_{i+2j}$ is identified with $c_i$.
    
    Fix one of the two colorings induced by $\C_{2j}$; without loss of generality, red is more frequent in $\C_{2j}$. Let $A(c_i, c_{i+1})$ be the smallest red arc. Define $\C_{2j-2} \triangleq \C_{2j} \setminus \{c_i, c_{i+1}\}$. In other words, $\C_{2j-2}$ is obtained from $\C_{2j}$ by changing the color of arc $A(c_i, c_{i+1})$ from red to blue. The resulting number of blue nodes in $\C_{2j-2}$ equals $\less(\C_{2j}) + d(c_i, c_{i+1})$.

    We first upper-bound $d(c_i, c_{i+1})$. Since the invariant holds for $\C_{2j}$, we have $\less(\C_{2j}) \geq n/2 - n/(2j)$, and therefore the number of red nodes in $\C_{2j}$ is at most $n/2 + n/(2j)$. Because $\C_{2j}$ contains $j$ red arcs and $A(c_i, c_{i+1})$ is the smallest among them, we obtain $d(c_i, c_{i+1}) \leq n/(2j) + n/(2j^2)$.

    If blue remains the less frequent color also in $\C_{2j-2}$, the invariant follows immediately, since $\less(\C_{2j-2}) = \less(\C_{2j}) + d(c_i, c_{i+1}) \geq \less(\C_{2j}) \geq n/2 - n/(2j) > n/2 - n/(2j-2)$.

    Otherwise, red becomes the less frequent color in $\C_{2j-2}$. Then
    \begin{align*}
        \less(\C_{2j-2}) 
        & = n - \less(\C_{2j}) - d(c_i, c_{i+1}) \\
        & \geq n - n/2 - n/(2j) - n/(2j^2) && \text{(as $\less(\C_{2j}) \leq n/2$)} \\
        & > n/2 - n/(2j-2).
    \end{align*}
    In either case, the invariant holds for $\C_{2j-2}$, and the construction is therefore complete.
\end{proof}


\subsection{Distance Between OPT and OFF (Green Sectors)}

We define $\Phi$ as the ``edit distance'' between \OPT and \OFF, representing the smallest set of nodes that must switch clusters to transform the partition of \OFF into that of \OPT.

To this end, we adapt the notion of an arc: a \emph{clockwise arc} $A^{\circ}(c_i, c_j)$ includes all nodes between $c_i$ and $c_j$ in the clockwise order (not necessarily along the shorter path connecting them). The symbol $\uplus$ denotes the disjoint union of sets.

\begin{definition}
    \label{def:potential}
    If $\CE{\OPT}=\CE{\OFF}$, then $\Phi = \emptyset$.
    Otherwise, let $c_1, c_2, \ldots, c_m$ denote the cut-edges from $\CE{\OPT} \setminus \CE{\OFF}$ listed in clockwise order. Define 
    \begin{itemize}
        \item $\Phi_0 \triangleq \AC(c_1, c_2) \uplus \AC(c_3, c_4) \uplus \ldots \uplus \AC(c_{m-1},c_m)$, and
        \item $\Phi_1 \triangleq \AC(c_m,c_1) \uplus \AC(c_2, c_3) \uplus \ldots \uplus \AC(c_{m-2},c_{m-1})$.
    \end{itemize}
    Then, $\Phi$ is the smaller of the two sets, with ties broken arbitrarily; cf.~\autoref{fig:potential}. 
\end{definition}

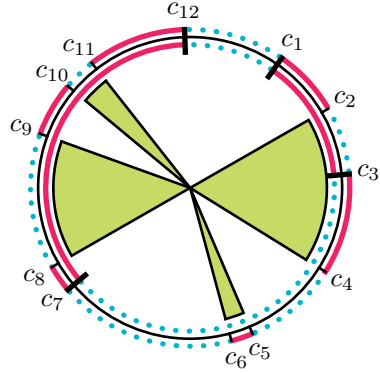
\begin{figure*}[t!]
    \centering
    \scalebox{1.0}{
    \begin{tikzpicture}[line cap=rect,line width=1pt]
    \filldraw [fill=white] (0,0) circle [radius=2cm];
    %
        \foreach \i/\j/\k in {-32/5/WildStrawberry,  6/30/Turquoise, 31/55/WildStrawberry, 56/92/Turquoise, -267/-232/WildStrawberry, -231/-220/Turquoise, -219/-200/WildStrawberry, -199/-150/Turquoise, -149/-140/WildStrawberry, -139/-75/Turquoise, -74/-67/WildStrawberry, -66/-33/Turquoise}
    {
        \ifthenelse{\equal{\k}{Turquoise}}{
            \draw[line width=2pt, color={\k}, line cap=round, dash pattern=on 0pt off 5pt]
                ([shift=(\i:2.1cm)]0,0) arc (\i:\j:2.1cm);
        }{
            \draw[line width=2pt, color={\k}]
                ([shift=(\i:2.1cm)]0,0) arc (\i:\j:2.1cm);
        }
    }
    \foreach \i/\j/\k in {5/55/WildStrawberry,  56/92/Turquoise, -267/-140/WildStrawberry, -139/4/Turquoise}
    {
        \ifthenelse{\equal{\k}{Turquoise}}{
            \draw[line width=2pt, color={\k}, line cap=round, dash pattern=on 0pt off 5pt]
                ([shift=(\i:1.9cm)]0,0) arc (\i:\j:1.9cm);
        }{
            \draw[line width=2pt, color={\k}]
                ([shift=(\i:1.9cm)]0,0) arc (\i:\j:1.9cm);
        }
    }
    %
    %
    \foreach \angle [count=\xi] in {55, 30, 5, -32, -67, -75, -140, -150, -200, -220, -232, -268}
    {
        \draw[line width=1pt] (\angle:2cm) -- (\angle:2.1cm);
        \node[font=\normalsize] at (\angle:2.36cm) {$c_{\xi}$};
    }
    \foreach \angle [count=\xi] in {55,  5,  -140,  -268}
    {
        \draw[line width=2pt] (\angle:1.8cm) -- (\angle:2.1cm);
    }
    \foreach \i/\j/\k in {30/-32/SpringGreen,  -67/-75/SpringGreen, -150/-200/SpringGreen, -220/-232/SpringGreen}
    {
    \draw[fill=\k] (0,0) -- (\i:1.8) arc (\i:\j:1.8);
        \draw[line width=1pt, color=black] (\j:0cm) -- (\j:1.8cm);
    }
    \end{tikzpicture}
    }
    \caption{An example configuration of \OPT and \OFF with $\CE{\OPT} = {c_1, c_2, \dots, c_{12}}$ and $\CE{\OFF}={c_1, c_3, c_7, c_{12}}$. Cut-edges of \OFF are shown with thicker lines. The coloring outside the circle corresponds to \OPT, while the coloring inside corresponds to \OFF. The set $\Phi$ consists of nodes indicated by green sectors. The endpoints of these sectors coincide with $\CE{\OPT} \setminus \CE{\OFF}$.}
    \label{fig:potential}
\end{figure*}

\begin{lemma}
    $\Phi$ is the minimum-cardinality set of nodes that \OFF has to recolor to make its coloring either identical to or completely opposite to the coloring of \OPT.
\end{lemma}

\begin{proof}
    First, observe that $\Phi_0$ and $\Phi_1$ are disjoint, and their union is the set of all nodes. Second, either $\Phi_0$ or $\Phi_1$ consists of the nodes on which the colorings of \OPT and \OFF differ. Consequently, the other set contains nodes on which the colorings of \OPT and \OFF coincide.
\end{proof}

As an extreme illustration, consider the case in which the coloring of \OFF is the ``inverse'' of the coloring of \OPT (nodes that are blue in the \OPT solution are red in \OFF, and vice versa). In this case, the colorings correspond to the same partitions, and $\Phi = \emptyset$.


\subsection{A Single Color Change of OPT: Definition of OFF}
\label{sec:single_recoloring}

We say that the state of an edge is \emph{flipped} if, due to some action of an algorithm, the edge becomes a cut-edge when it was not one, or ceases to be a cut-edge when it was one.

\begin{observation}
    \label{obs:arc_coloring}
    Fix two edges $e_1$ and $e_2$. Flipping their states is equivalent to recoloring all nodes on arc~$A(e_1,e_2)$ (and only them).
\end{observation}

Recall that to serve a single request, \OPT may recolor several nodes. We refine its recolorings per request to subsequences of recolorings of single nodes. Using the definition of set $\Phi$, we now specify how \OFF responds to a single color change performed by \OPT in step~$t$. Let $w$ denote this node.

By \autoref{obs:arc_coloring}, this color change flips the states of the two edges adjacent to $w$. Consequently, \OPT either introduces two new adjacent cut-edges, deletes two adjacent cut-edges, or shifts an existing cut-edge by one position. We consider these three cases and specify how \OFF responds to each, both in terms of recoloring and the resulting cut-edge modifications.

In the description below, $\CE{\OFF^{t-1}}$, $\CE{\OPT^{t-1}}$, and $\Phi^{t-1}$ denote the sets of cut-edges of \OFF and \OPT, and the value of $\Phi$, respectively, immediately before step $t$. Recall that $\CE{\OFF^{t-1}} \subseteq \CE{\OPT^{t-1}}$ by the cut-edge invariants. 

\begin{description}
    \item[OPT shifts a cut-edge $c_i$.]
    If $c_i \in \CE{\OFF^{t-1}}$, then \OFF also recolors $w$, thereby shifting~$c_i$ as well; otherwise, it does nothing. 

    \item[OPT introduces two new adjacent cut-edges $c_i$ and $c_j$.]
    If $|\CE{\OFF^{t-1}}| \leq 2 (k-1)$, then \OFF also recolors $w$, thereby adding $c_i$ and $c_j$ to $\CE{\OFF}$; otherwise, it does nothing.

    \item[OPT removes two adjacent cut-edges $c_i$ and $c_j$.] 
    We consider four sub-cases:
    \begin{itemize}
        \item If $c_i, c_j \notin \CE{\OFF^{t-1}}$, then \OFF takes no action.

        \item If $c_i, c_j \in \CE{\OFF^{t-1}}$ and $\CE{\OFF^{t-1}} = \CE{\OPT^{t-1}}$, then \OFF also recolors~$w$, thereby removing $c_i$ and $c_j$ from $\CE{\OFF}$.

        \item If $c_i, c_j \in \CE{\OFF^{t-1}}$ and $\CE{\OPT^{t-1}} \setminus \CE{\OFF^{t-1}} \neq \emptyset$, then \OFF first recolors~$w$, thereby removing $c_i$ and $c_j$. Next, \OFF selects an arbitrary arc $A(c_x, c_y) \in \Phi^{t-1}$ (guaranteed by~\autoref{def:potential}) and recolors all nodes on~$A(c_x, c_y)$. By \autoref{obs:arc_coloring}, this operation adds $c_x$ and $c_y$ to $\CE{\OFF}$.

        \item In the final sub-case, exactly one of $\{c_i, c_j\}$, say $c_i$, does not belong to $\CE{\OFF^{t-1}}$. Since $c_i \in \CE{\OPT^{t-1}} \setminus \CE{\OFF^{t-1}}$, \autoref{def:potential} implies the existence of an~arc $A(c_i, c_y) \in \Phi^{t-1}$. In this situation, \OFF recolors all nodes on $A(c_j, c_y)$. By \autoref{obs:arc_coloring}, this operation removes $c_j$ from $\CE{\OFF}$ and adds $c_y$ to it.
    \end{itemize}
\end{description}

It is straightforward to verify that in each of the above cases, \OFF preserves the cut-edge invariants. In particular, every cut-edge added by \OFF is present in $\CE{\OPT^t}$. Hence, it suffices to analyze the recoloring cost of \OFF, which will be done in the next section.


\subsection{A Single Color Change of OPT: Cost Analysis}
\label{sec:change}

We analyze a single step $t$ in which \OPT recolors a node $w$.

Let $\Phi^{t-1}$ and $\Phi^t$ denote the values of $\Phi$ immediately before and after step $t$. Let $\Phi^t_*$ be equal to~$\Phi^t_0$ if $\Phi^{t-1} = \Phi^{t-1}_0$, and to $\Phi^t_1$ otherwise. That is, $\Phi^t_*$ represents the ``continuation'' of $\Phi^{t-1}$, disregarding a possible switch between $\Phi_0$ and $\Phi_1$. Clearly $|\Phi^t| \leq  |\Phi^t_*|$.

We define the potential function as $\phi \triangleq |\Phi|$. Let $\Delta^t \phi = \phi^t - \phi^{t-1}$, and let $\Delta^t_* \phi = \phi^t_* - \phi^{t-1}$.

In the following three lemmas, we analyze the possible changes performed by \OPT. In each lemma, we show that $\MC{\OFF^t} + \Delta^t_* \phi \leq 1$. Since $\phi^t \leq \phi^t_*$, it follows that $\MC{\OFF^t} + \Delta^t \phi \leq \MC{\OFF^t} + \Delta^t_* \phi \leq 1 = \MC{\OPT^t}$. The symbol $\oplus$ denotes the symmetric difference of sets.

\begin{lemma}
    \label{lem:case1}
    If \OPT shifts a cut-edge, then $\MC{\OFF^t} + \Delta^t_* \phi \leq 1$.
\end{lemma}

\begin{proof}
    Let $c_i$ denote the shifted cut-edge. 
    \begin{itemize}
        \item If $c_i \in \CE{\OFF^{t-1}}$, then \OFF recolors $w$, paying a cost of $1$. As a result, \OFF reproduces the shift of a~cut-edge (cf.~\autoref{fig:case3}a and \ref{fig:case3}b). This leaves $\Phi$ unchanged, and thus $\Delta^t_* \phi = 0$.
    
        \item Otherwise, $c_i \notin \CE{\OFF^{t-1}}$, and then \OFF does nothing, thus its cost is $0$, cf.~\autoref{fig:case3}c. In this case, $\Phi^t_* = \Phi^{t-1} \oplus \{w\}$, and therefore $\Delta^t_*\phi = \pm 1$.
        \qedhere    
    \end{itemize}
\end{proof}

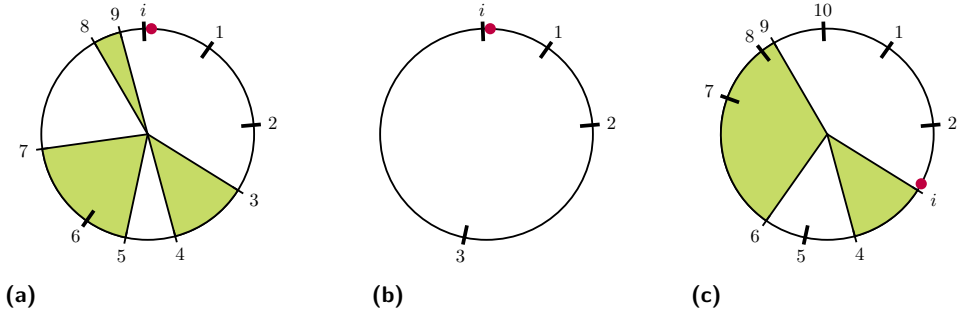
\begin{figure*}[ht]
    \centering
    \begin{subfigure}[t]{0.3\textwidth}
        \scalebox{0.7}{
        \begin{tikzpicture}[line cap=rect,line width=1pt]
        \filldraw [fill=white] (0,0) circle [radius=2cm];
        \foreach \i/\j/\k in {-32/-75/SpringGreen,  -102/-172/SpringGreen, -240/-255/SpringGreen}
        {
        \draw[fill=\k] (0,0) -- (\i:2) arc (\i:\j:2);
         \draw[line width=1pt, color=black] (\j:0cm) -- (\j:2cm);
        }
        \foreach \angle [count=\xi] in {55, 5, -32, -75, -102, -125, -172, -240, -255, -268}
        {
          \draw[line width=1pt] (\angle:1.9cm) -- (\angle:2.1cm);
          \node[font=\normalsize] at (\angle:2.36cm) {\ifnum\xi=10 $i$\else\xi\fi};
        }
         \fill[purple]  ++(-272:2cm) circle[radius=3pt];
       \foreach \angle in {55, 5, -125, -268}
          \draw[line width=2pt, color=black] (\angle:1.8cm) -- (\angle:2.1cm);
        \end{tikzpicture}
        }   
        \caption{}
    \end{subfigure}%
    \hspace{0.2in}
    \begin{subfigure}[t]{0.3\textwidth}
        \scalebox{0.7}{
        \begin{tikzpicture}[line cap=rect,line width=1pt]
        \filldraw [fill=white] (0,0) circle [radius=2cm];
        \foreach \angle [count=\xi] in {55, 5, -102, -268}
        {
          \draw[line width=1pt] (\angle:1.9cm) -- (\angle:2.1cm);
          \node[font=\normalsize] at (\angle:2.36cm) {\ifnum\xi=4 $i$\else\xi\fi};
        }
         \fill[purple]  ++(-272:2cm) circle[radius=3pt];
        \foreach \angle in {55, 5, -102, -268}
          \draw[line width=2pt, color=black] (\angle:1.8cm) -- (\angle:2.1cm);
        \end{tikzpicture}
        }
    \caption{}
    \end{subfigure}%
    \begin{subfigure}[t]{0.3\textwidth}
        \scalebox{0.7}{
        \begin{tikzpicture}[line cap=rect,line width=1pt]
        \filldraw [fill=white] (0,0) circle [radius=2cm];
        \foreach \i/\j/\k in {-32/-75/SpringGreen,  -125/-240/SpringGreen}
        {
        \draw[fill=\k] (0,0) -- (\i:2) arc (\i:\j:2);
         \draw[line width=1pt, color=black] (\j:0cm) -- (\j:2cm);
        }
        \foreach \angle [count=\xi] in {55, 5, -32, -75, -102, -125, -200, -232, -240, -268}
        {
          \draw[line width=1pt] (\angle:1.9cm) -- (\angle:2.1cm);
          \node[font=\normalsize] at (\angle:2.36cm) {\ifnum\xi=3 $i$\else\xi\fi};
        }
         \fill[purple]  ++(-28:2cm) circle[radius=3pt];
       \foreach \angle in {55, 5, -102, -232, -200, -268}
          \draw[line width=2pt, color=black] (\angle:1.8cm) -- (\angle:2.1cm);
        
    \end{tikzpicture}
    }
    \caption{}
    \end{subfigure}%
    \caption{$\Phi$ right before \OPT shifts a cut-edge. For simplicity, we label cut-edges by~$i$ instead of~$c_i$. The purple dot represents the node $w$ recolored by \OPT, thick cut-edges belong to $\CE{\OFF^{t-1}}$, and nodes in $\Phi^{t-1}$ are marked with green sectors.}
    \label{fig:case3}
\end{figure*}

\begin{lemma}
    \label{lem:case2}
    If \OPT adds two new adjacent cut-edges, then $\MC{\OFF^t} + \Delta^t_* \phi \leq 1$.
\end{lemma}

\begin{proof}
    We distinguish two cases based on the cardinality of $\CE{\OFF^{t-1}}$.
    \begin{itemize}
        \item If $|\CE{\OFF^{t-1}}| = 2k$, \OFF takes no action, and hence incurs no cost, cf.~\autoref{fig:case1}a and~\ref{fig:case1}b. In this case, $\Phi^t_* = \Phi^{t-1} \oplus \{w\}$, and thus $\Delta^t_* \phi = \pm 1$.
        
        \item Otherwise $|\CE{\OFF^{t-1}}| \leq 2 (k-1)$. Then, \OFF recolors $w$ paying a cost of $1$, cf.~\autoref{fig:case1}c. In this case, $\Phi$ remains unchanged, and thus $\Delta^t_* \phi = 0$.
    \end{itemize}
    \qedhere
\end{proof}

\begin{figure*}[ht]
    \centering
    \begin{subfigure}[t]{0.3\textwidth}
        \scalebox{0.7}{
        \begin{tikzpicture}[line cap=rect,line width=1pt]
        \filldraw [fill=white] (0,0) circle [radius=2cm];
        \foreach \i/\j/\k in {-32/-75/SpringGreen,  -125/-200/SpringGreen}
        {
        \draw[fill=\k] (0,0) -- (\i:2) arc (\i:\j:2);
         \draw[line width=1pt, color=black] (\j:0cm) -- (\j:2cm);
        }
        \foreach \angle [count=\xi] in {55, 5, -32, -75, -102, -125, -200, -232, -240, -268}
        {
          \draw[line width=1pt] (\angle:1.9cm) -- (\angle:2.1cm);
          \node[font=\normalsize] at (\angle:2.36cm) {\xi};
        }
        
        \fill[purple]  ++(-210:2cm) circle[radius=3pt];
        \foreach \angle in {55, 5, -102, -232, -240, -268}
          \draw[line width=2pt, color=black] (\angle:1.8cm) -- (\angle:2.1cm);
        \end{tikzpicture}
        }
    \caption{}
    \end{subfigure}%
    \begin{subfigure}[t]{0.3\textwidth}
        \scalebox{0.7}{
        \begin{tikzpicture}[line cap=rect,line width=1pt]
        \filldraw [fill=white] (0,0) circle [radius=2cm];
        \foreach \i/\j/\k in {30/-32/SpringGreen,  -67/-75/SpringGreen, -150/-200/SpringGreen, -220/-232/SpringGreen}
        {
        \draw[fill=\k] (0,0) -- (\i:2) arc (\i:\j:2);
         \draw[line width=1pt, color=black] (\j:0cm) -- (\j:2cm);
        }
        \foreach \angle [count=\xi] in {55, 30, 5, -32, -67, -75, -125, -150, -200, -220, -232, -268}
        {
          \draw[line width=1pt] (\angle:1.9cm) -- (\angle:2.1cm);
          \node[font=\normalsize] at (\angle:2.36cm) {\xi};
        }
        \fill[purple]  ++(-10:2cm) circle[radius=3pt];
        \foreach \angle in {55, 5, -125, -268}
          \draw[line width=2pt, color=black] (\angle:1.8cm) -- (\angle:2.1cm);
        \end{tikzpicture}
        }
    \caption{}
    \end{subfigure}%
    \hspace{0.1in}
    \begin{subfigure}[t]{0.3\textwidth}
        \scalebox{0.7}{
        \begin{tikzpicture}[line cap=rect,line width=1pt]
        \filldraw [fill=white] (0,0) circle [radius=2cm];
        
        \foreach \angle [count=\xi] in {55, 5, -102, -268}
        {
          \draw[line width=1pt] (\angle:1.9cm) -- (\angle:2.1cm);
          \node[font=\normalsize] at (\angle:2.36cm) {\xi};
        }
         \fill[purple]  ++(-236:2cm) circle[radius=3pt];
        \foreach \angle in {55, 5, -102, -268}
          \draw[line width=2pt, color=black] (\angle:1.8cm) -- (\angle:2.1cm);
        \end{tikzpicture}
        }
    \caption{}
    \end{subfigure}
    \caption{$\Phi$ right before \OPT introduces two new cut-edges. Notation as in \autoref{fig:case3}.}
    \label{fig:case1}
\end{figure*}
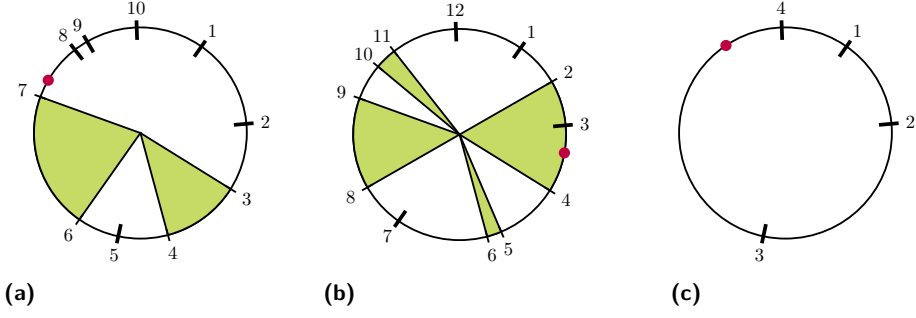

\begin{lemma}
    \label{lem:case3}
    If \OPT removes two adjacent cut-edges, then $\MC{\OFF^t} + \Delta^t_* \phi \leq 1$.
\end{lemma}

\begin{proof}
    Let $c_i, c_j$ denote the cut-edges removed by \OPT. 
    \begin{itemize}
        \item If $c_i, c_j \notin \CE{\OFF^{t-1}}$, \OFF takes no action and therefore incurs no cost, cf.~\autoref{fig:case2}a. In this case, $\Phi^t_* = \Phi^{t-1} \oplus \{w\}$, implying $\Delta^t_* \phi = \pm 1$.

        \item If $c_i, c_j \in \CE{\OFF^{t-1}}$ and $\CE{\OFF^{t-1}} = \CE{\OPT^{t-1}}$, then \OFF recolors $w$ paying a cost of~$1$. In this case, $\Phi$ remains unchanged, and hence $\Delta^t_* \phi = 0$.

        \item In the next case, $c_i, c_j \in \CE{\OFF^{t-1}}$ and $\CE{\OPT^{t-1}} \setminus \CE{\OFF^{t-1}} \neq \emptyset$. Let $A(c_x, c_y) \in \Phi^{t-1}$ be the arc selected by \OFF; see~\autoref{fig:case2}b, where this arc is marked with a star. We first examine the effect of both \OPT and \OFF recoloring $w$, and then subsequent recoloring of $A(c_x,c_y)$ by \OFF. When both recolor $w$, \OFF pays a cost of $1$, and $\Phi$ remains unchanged. Next, when \OFF recolors $A(c_x, c_y)$, it pays and additional cost of $d(c_x,c_y)$, and $\Phi^t_* = \Phi^{t-1} \oplus A(c_x, c_y)$. Since $A(c_x, c_y) \subseteq \Phi^{t-1}$, we have $\Phi^t_* = \Phi^{t-1} \setminus A(c_x, c_y)$. Overall, $\MC{\OFF^t} = 1 + d(c_x, c_y)$ and $\Delta \phi = - d(c_x, c_y)$.

        \item In the final case, $c_i \notin \CE{\OFF^{t-1}}$ and $c_j \in \CE{\OFF^{t-1}}$. By \autoref{def:potential}, there exists an arc $A(c_i, c_y) \subseteq \Phi^{t-1}$, marked with a star in~\autoref{fig:case2}c and \ref{fig:case2}d. \OFF recolors all nodes in $A(c_j, c_y)$ (note that $A(c_j, c_y) = A(c_i, c_y) \oplus \{w\}$), incurring cost of at most $d(c_i, c_y)+1$. In this case, $\Phi^t_* = \Phi^{t-1} \oplus A(c_j, c_y) \oplus \{w\} = \Phi^{t-1} \oplus A(c_i, c_y)$. Since $A(c_i, c_y) \subseteq \Phi^{t-1}$, we obtain $\Phi^t_* = \Phi^{t-1} \setminus A(c_i, c_y)$, and hence $\Delta^t_* \phi = -d(c_i,c_y)$. 
    \qedhere
    \end{itemize}
\end{proof}

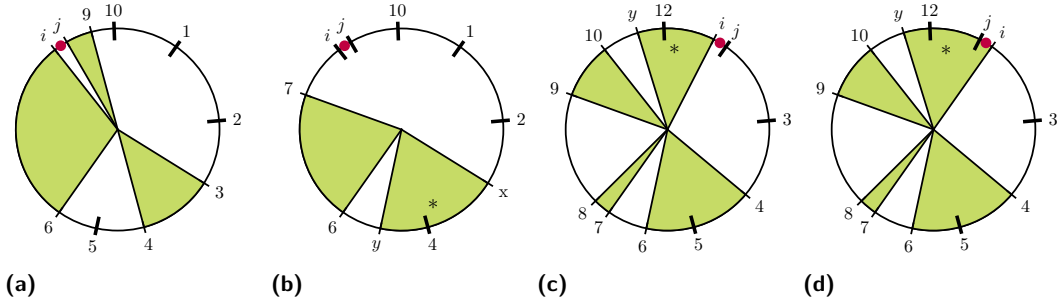
\begin{figure*}[ht]
    \centering
    \begin{subfigure}[t]{0.25\textwidth}
        \scalebox{0.67}{
        \begin{tikzpicture}[line cap=rect,line width=1pt]
        \filldraw [fill=white] (0,0) circle [radius=2cm];
        \foreach \i/\j/\k in {-32/-75/SpringGreen,  -125/-232/SpringGreen, -240/-255/SpringGreen}
        {
        \draw[fill=\k] (0,0) -- (\i:2) arc (\i:\j:2);
         \draw[line width=1pt, color=black] (\j:0cm) -- (\j:2cm);
        }
        \foreach \angle [count=\xi] in {55, 5, -32, -75, -102, -125, -232, -240, -255, -268}
        {
        \draw[line width=1pt] (\angle:1.9cm) -- (\angle:2.1cm);
        \node[font=\normalsize] at (\angle:2.36cm)
         {\ifnum\xi=7 $i$\else\ifnum\xi=8 $j$\else\xi\fi\fi};
        }
         \fill[purple]  ++(-236:2cm) circle[radius=3pt];
        \foreach \angle in {55, 5, -102, -268}
          \draw[line width=2pt, color=black] (\angle:1.8cm) -- (\angle:2.1cm);
        \end{tikzpicture}
        }
  \caption{}
    \end{subfigure}%
    \begin{subfigure}[t]{0.25\textwidth}
        \scalebox{0.67}{
        \begin{tikzpicture}[line cap=rect,line width=1pt]
        \filldraw [fill=white] (0,0) circle [radius=2cm];
        \foreach \i/\j/\k in {-32/-102/SpringGreen,  -125/-200/SpringGreen}
        {
        \draw[fill=\k] (0,0) -- (\i:2) arc (\i:\j:2);
         \draw[line width=1pt, color=black] (\j:0cm) -- (\j:2cm);
        }
        \foreach \angle [count=\xi] in {55, 5, -32, -75, -102, -125, -200, -232, -240, -268}
        {
          \draw[line width=1pt] (\angle:1.9cm) -- (\angle:2.1cm);
          \node[font=\normalsize] at (\angle:2.36cm)
            {\ifnum\xi=3 x\else\ifnum\xi=5 $y$\else\ifnum\xi=8 $i$\else\ifnum\xi=9 $j$\else\xi\fi\fi\fi\fi};
        }
         \fill[purple]  ++(-236:2cm) circle[radius=3pt];
         \node[] at (-67:1.6cm) {\Large{$*$}};
       \foreach \angle in {55, 5, -75, -232, -240, -268}
          \draw[line width=2pt, color=black] (\angle:1.8cm) -- (\angle:2.1cm);
        \end{tikzpicture}
        }
    \caption{}
    \end{subfigure}%
    \begin{subfigure}[t]{0.25\textwidth}
        \scalebox{0.67}{
        \begin{tikzpicture}[line cap=rect,line width=1pt]
        \filldraw [fill=white] (0,0) circle [radius=2cm];
        \foreach \i/\j/\k in {-40/-102/SpringGreen,  -125/-135/SpringGreen, -200/-232/SpringGreen, 107/63/SpringGreen}
        {
        \draw[fill=\k] (0,0) -- (\i:2) arc (\i:\j:2);
         \draw[line width=1pt, color=black] (\j:0cm) -- (\j:2cm);
        }
        \foreach \angle [count=\xi] in {63, 55, 5, -40, -75, -102, -125, -135, -200, -232, -253, -268}
        {
          \draw[line width=1pt] (\angle:1.9cm) -- (\angle:2.1cm);
          \node[font=\normalsize] at (\angle:2.36cm)
            {\ifnum\xi=1 $i$\else\ifnum\xi=2 $j$\else\ifnum\xi=11 $y$\else\xi\fi\fi\fi};
        }
         \fill[purple]  ++(59:2cm) circle[radius=3pt];
         \node[] at (85:1.6cm) {\Large{$*$}};
        \foreach \angle in {55, 5, -75, -268 }
          \draw[line width=2pt, color=black] (\angle:1.8cm) -- (\angle:2.1cm);
        \end{tikzpicture}
        }
        \caption{}
    \end{subfigure}%
    \begin{subfigure}[t]{0.25\textwidth}
        \scalebox{0.67}{
        \begin{tikzpicture}[line cap=rect,line width=1pt]
        \filldraw [fill=white] (0,0) circle [radius=2cm];
        \foreach \i/\j/\k in {-40/-102/SpringGreen,  -125/-135/SpringGreen, -200/-232/SpringGreen, 107/55/SpringGreen}
        {
        \draw[fill=\k] (0,0) -- (\i:2) arc (\i:\j:2);
         \draw[line width=1pt, color=black] (\j:0cm) -- (\j:2cm);
        }
        \foreach \angle [count=\xi] in {63, 55, 5, -40, -75, -102, -125, -135, -200, -232, -253, -268}
        {
          \draw[line width=1pt] (\angle:1.9cm) -- (\angle:2.1cm);
          \node[font=\normalsize] at (\angle:2.36cm)
            {\ifnum\xi=1 $j$\else\ifnum\xi=2 $i$\else\ifnum\xi=11 $y$\else\xi\fi\fi\fi};
        }
         \fill[purple]  ++(59:2cm) circle[radius=3pt];
         \node[] at (81:1.6cm) {\Large{$*$}};
        \foreach \angle in {63, 5, -75, -268 }
          \draw[line width=2pt, color=black] (\angle:1.8cm) -- (\angle:2.1cm);
        \end{tikzpicture}
        }
    \caption{}
    \end{subfigure}%
    \caption{$\Phi$ right before \OPT deletes two adjacent cut-edges. Notation as in \autoref{fig:case3}.}
    \label{fig:case2}
\end{figure*} 

From \autoref{lem:case1}, \autoref{lem:case2} and \autoref{lem:case3} together with $\Delta^t \phi \le \Delta^t_* \phi$ and $\MC{\OPT^t} = 1$, we immediately obtain the following bound.

\begin{corollary}
    \label{cor:step}
    For every step $t$, it holds that $\MC{\OFF^t} + \Delta^t \phi \leq \MC{\OPT^t}$. 
\end{corollary}


\subsection{Approximation Ratio of OFF} 

Recall that before processing the first request, \OFF performs the global rebalancing procedure described in \autoref{lem:rebalancing}. This guarantees that \OFF satisfies the cut-edge invariants and remains $(1 + 1/k)$-balanced (and therefore also $\alpha$-balanced, since $\alpha = 3/2 + 1/k$). The rebalancing procedure is invoked whenever \OFF ceased to be $\alpha$-balanced after applying the changes induced by \OPT recolorings. Consequently, $\OFF \in \A^k_\alpha$, and it remains to bound its total cost and compare it with that of \OPT.

We begin with the following auxiliary lemma.

\begin{lemma}
\label{lem:lb_potential}
    Whenever \OPT is $1$-balanced, $\less(\CE{\OFF}) + \phi \geq n/2$.
\end{lemma}

\begin{proof}
    Let $z_{\mathrm{RR}}$, $z_{\mathrm{RB}}$, $z_{\mathrm{BR}}$ and $z_{\mathrm{BB}}$ denote the number of nodes colored red by both \OPT and \OFF, red by \OPT and blue by \OFF, blue by \OPT and red by \OFF, and blue by both \OPT and \OFF, respectively. By the definitions of $\phi$ and $\less(\CE{\OFF})$, we obtain 
    \begin{align*}
        \less(\CE{\OFF}) + \phi
        & = \min \{z_{\mathrm{RR}} + z_{\mathrm{BR}},  z_{\mathrm{BB}} + z_{\mathrm{RB}} \} 
        + \min \{z_{\mathrm{RB}} + z_{\mathrm{BR}}, z_{\mathrm{RR}} + z_{\mathrm{BB}} \}, \\
        & \geq \min \{ z_{\mathrm{RB}} + z_{\mathrm{RR}}, z_{\mathrm{BR}} + z_{\mathrm{BB}} \} \\
        & = \min\{n/2, n/2\} = n/2,
    \end{align*}
    where the penultimate equality holds because \OPT is $1$-balanced.
\end{proof}

In the following, we fix an input sequence $\sigma$ and run both $\OPT$ and $\OFF$ on it. This induces a sequence of steps, and \OFF performs global rebalancings between some of them. These rebalancings partition the sequence of steps into \emph{phases}; the last phase may not necessarily end with a global rebalancing.

For any phase $p$, let $\MC{\OFF}(p)$ denote the recoloring cost of \OFF during $p$, and let $\Delta_p \phi$ denote the change in potential $\phi$ during the same phase, both measured excluding the rebalancing. The following lemma states that the amortized cost incurred by \OFF in~$p$ is appropriately large, which later enables us to upper-bound the cost of the rebalancing procedure.

\begin{lemma}
    \label{lem:opt_lb}
    For any phase $p$, it holds that $\MC{\OFF}(p) + \Delta_p \phi \geq n/(2k)$. 
\end{lemma}

\begin{proof}
    Let $\C_b$ and $\C_e$ denote the sets of cut-edges of \OFF at the beginning and at the end of phase $p$, respectively; the latter is this set immediately before the global rebalancing. Then, 
    \begin{align*}
        \less(\C_e) & < n - \alpha \cdot (n/2) = n/4 - n/(2k),
            && \text{(as rebalancing was triggered)} \\
        \less(\C_b) & \geq n/2 - n/(2k).
            && \text{(by \autoref{lem:rebalancing})}
    \end{align*}
    Next, we show that 
    \begin{equation}
        \label{eq:mc_off_bound}
        \MC{\OFF}(p) \geq \less(\C_b) - \less(\C_e). 
    \end{equation}
    To this end, assume without loss of generality that red is the less frequent color in $\C_e$. If red is also the less frequent color in $\C_b$, then \eqref{eq:mc_off_bound} follows immediately. Otherwise, blue was less frequent in $\C_b$, and then there exist a moment within~$p$ when exactly $n/2$ nodes were colored red. Hence, $\MC{\OFF}(p) \geq n/2 - \less(\C_e) \geq \less(\C_b) - \less(\C_e)$.

    We now derive a lower bound on $\Delta_p \phi$. The last step of each phase corresponds to the final recoloring performed by \OPT when serving some request. Consequently, \OPT is $1$-balanced at that time. By \autoref{lem:lb_potential}, the potential value immediately before rebalancing is at least $n/2 - \less(\C_e)$. Since the potential at the beginning of $p$ is at most $n/2$, it follows that
    \begin{equation}
        \label{eq:delta_phi_bound}
        \Delta_p\phi \geq - \less(\C_e).
    \end{equation}
    Combining \eqref{eq:mc_off_bound} and \eqref{eq:delta_phi_bound} yields
    $\MC{\OFF}(p) + \Delta_p\phi 
        \geq \less(\C_b) - 2 \cdot \less(\C_e) 
        \geq n/(2k)$.
\end{proof}

We can finally prove \autoref{thm:offline}, restated below. 

\offline*

\begin{proof}
    As established above, $\OFF \in \A^k_\alpha$, so it remains to bound its total cost. Fix an input sequence~$\sigma$, and let it consist of $h$ phases $p_1, p_2, \ldots, p_h$.

    For each phase $p$, let $\OFF(\tp)$ and $\Delta_{\tp} \phi$, be the cost of \OFF within $p$ with rebalancing, and the change of potential within $p$ including the rebalancing. Let $\OPT(p)$ denote the cost of \OPT in phase $p$. The rebalancing performed at the end of~$p$ incurs a cost of at most~$n$, and the corresponding change in potential is at most~$n/2$. Hence, for each phase~$p$,
    \begin{align*}
        \OFF&(\tp)  + \Delta_{\tp}\phi \\
        & \leq \HC{\OFF}(p) + \MC{\OFF}(p) + \Delta_{p}\phi + (3/2) \cdot n \\
        & \leq \HC{\OFF}(p) + (3k + 1) \cdot (\MC{\OFF}(p) + \Delta_{p}\phi) 
            && \text{(by \autoref{lem:opt_lb})} \\
        & \leq \HC{\OFF}(p) + (3k + 1) \cdot \MC{\OPT}(p)
            && \text{(by \autoref{cor:step} over steps in $p$)} \\
        & \leq \HC{\OPT}(p) + (3k + 1) \cdot \MC{\OPT}(p)
            && \text{(as $\CE{\OFF} \subseteq \CE{\OPT}$)} \\
        & \leq (3k + 1) \cdot \OPT(p).
    \end{align*}
    The initial rebalancing performed by \OFF at the beginning of~$\sigma$ incurs a cost of at most~$n$, and the corresponding change in potential is at most~$n/2$. Therefore,
    \begin{align*}
        \OFF(\sigma) + \Delta_{\sigma}\phi 
        & = (3/2) \cdot n + \textstyle \sum_{i=1}^h (\OFF(\tp_i) + \Delta_{\tp_i}\phi)  \\
        & = (3/2) \cdot n + \textstyle \sum_{i=1}^h (3k+1) \cdot \OPT(p_i) \\
        & = O(k) \cdot \OPT(\sigma) + O(n).
    \end{align*}
    Since the potential $\phi$ is initially~$0$ and always non-negative, we have $\Delta_{\sigma}\phi \ge 0$, and thus the theorem follows.
\end{proof}


\section{Online algorithm}
\label{sec:online}

We first recall the definition of a metrical task system (MTS). 
An MTS is defined on a metric space $(\mathcal{S}, d)$, where $\mathcal{S}$ is a set of states and $d$ specifies the distance between any two states. There is an initial state $x_0 \in S$. 

The request sequence $\langle c_t : \mathcal{S} \rightarrow \mathbb{R^+}, t \geq 1 \rangle$ consists of cost functions $c_t$, and an algorithm has to produce a sequence of states $\langle x_t \in \mathcal{S}, t \geq 1 \rangle$. The total cost of such a solution is $\sum_t c_t(x_t)+d(x_{t-1}, x_t)$.

\online*

\begin{proof}
    We first argue that the online bisection problem on $n$ nodes, restricted to algorithms from class $\A^k_\alpha$, forms a metrical task system. Each state in $\mathcal{S}$ is a set of cut-edges~$\C$ such that $|\C| \leq 2k$ and $\C$ is $\alpha$-balanced. Before serving the input sequence, every algorithm from class $\A^k_\alpha$ performs global rebalancing using the procedure from \autoref{lem:rebalancing}. This corresponds to selecting the initial state $x_0 \in \mathcal{S}$. The cost of this initial rebalancing can be neglected, as it is identical for all algorithms in $\A^k_\alpha$. The distance $d$ between two states is defined as the minimum number of nodes that have to be recolored to change one state into the other. The value of $c_t(x) \in \{0,1\}$ is the hit cost incurred if an~algorithm state is~$x$. 

    By the result of Bubeck et al.~\cite{BuCoLL19}, there exists a randomized $O(\log^2 |\mathcal{S}|)$-competitive algorithm for every MTS. Let \OB be their algorithm for the MTS instance described above. In this case, $|\mathcal{S}| = \sum_{i=1}^k \binom{n}{2i} \leq k \cdot  \binom{n}{2k} \leq k \cdot n^{2k}$, and thus $\log^2 |\mathcal{S}| = O(k^2 \cdot \log^2 n)$. Clearly, $\OB \in \A^k_\alpha$, and for every (offline) algorithm $\OFF \in \A^k_\alpha$ and every input~$\sigma$, it holds that $\E[\OB(\sigma)] \leq O(k^2 \cdot \log^2 n) \cdot \OFF(\sigma) + \beta$, where $\beta$ is a parameter depending on $\mathcal{S}$, but independent of $\sigma$.
\end{proof}

\section{Conclusions}

In this paper, we designed a randomized $O(\eps^{-3} \cdot \log^2 n)$-competitive algorithm for the online bisection problem with ring demands, requiring clusters of size $3/4+\eps$. No super-constant lower bound is known for this variant; the current general randomized lower bound of $\Omega(\log n)$ does not appear to be adaptable to the ring demands. A natural direction for future work is to close the gap between upper and lower bounds and, perhaps more importantly, to reduce the required augmentation factor to $1+\eps$.

One approach that could bring progress toward the general case of online bisection is to study variants in which requests are restricted to a set $Q$ different from the cycle. For instance, the case where $Q$ is a tree seems to require substantially different techniques.


\bibliography{references.bib}


\end{document}